\documentclass[12pt]{article}
\title{Suslin-Rao Codes : the bridge between pure  and applied}
\author{
\text{Nitin Kenjale} \href{https://orcid.org}{\small (ORCID: 0009-0004-2083-910X)} \\ 
\small 
K. J. Somaiya Institute of Technology, Sion,\\
\small
Mumbai-400022, India \\ \small \&\\
\text{Anuradha S. Garge} \href{https://orcid.org}{\small (ORCID: 0000-0001-9276-3985)} \\[1ex]
\small
Department of Mathematics, \\
\small 
University of Mumbai, Mumbai- 400098, India.} 
\date{{\small \today}} 
\usepackage{amsmath,amssymb,amsfonts}
\usepackage{amsthm}
\usepackage{graphicx,latexsym}
\usepackage{tikz} 
\usepackage{amscd, amsmath,amssymb,amsfonts}
\usepackage{amsthm} 
\usepackage{ulem}
\usepackage{graphicx}
\usepackage{amscd, amsmath, amssymb, amsfonts}
\usepackage{amsthm}
\usepackage{ulem}
\usepackage{tabularx}

\newtheorem{thm}{Theorem}[section]
\theoremstyle{definition}

\newtheorem{dfn}{Definition}[section]
\theoremstyle{remark}

\theoremstyle{plain}

\newtheorem{prop}[thm]{Proposition}
\usepackage[pdftex]{hyperref}
\newtheorem{rem}[thm]{Remark}

\theoremstyle{plain} 
\newtheorem{lemma}{Lemma}[section]

\def\B{{\operatorname{B}}}            
\def\T{{\operatorname{T}}}            
 
\begin{document}
\maketitle{}
\begin{abstract}
In the year $1977,$ while developing the theory of modules arising in the context of Serre's problem on projective modules in the paper ``On Stably Free Modules" (see  \cite{AS}), A. A. Suslin constructed matrices inductively using unimodular rows. Our aim is to connect these matrices  to the current and applicable field of coding theory. To this end, we construct Suslin-Rao codes, by using 
all one vectors over the finite field $\mathbb{F}_2$. 
The generator matrix of the proposed codes takes the form $G_r = [I_{2^{r-1}} \mid \mathcal{S}_{r-1}]$, where $\mathcal{S}_{r-1}$ is a recursively defined Suslin matrix and $I_{2^{r-1}}$ is the identity matrix. For each positive integer $r$, the constructed code has parameters $[n = 2^r,\, k = 2^{r-1},\, d]$, where $d = r$ or $r+1$ depending on whether $r$ is even or odd respectively. 
Comparative analysis with standard families such as Hamming, Simplex and BCH codes shows that the Suslin-Rao codes achieve better performance than Hamming and Simplex codes in terms of code rate and error correction capability, though BCH codes may outperform them for large code lengths. Bit error rate (BER) versus signal-to-noise ratio (SNR) simulations further support the practical applicability of these codes in noisy channels. The recursive nature of the construction allows efficient encoder design, making the proposed codes attractive for real-time communication systems and hardware  implementation.
\end{abstract}

Generator Matrix, Minimum Distance, Finite Fields.
\vspace{2 mm}


\section{Introduction}
A. A. Suslin (see \cite{AS}) introduced matrices, in his work on projective modules, particularly in connection with the Serre's conjecture, later known as the Quillen-Suslin theorem (see \cite{VA}). Today we know them as Suslin matrices. These matrices arise naturally in the study of unimodular rows and provide explicit algebraic constructions. Further Ravi Rao and Selby Jose studied Suslin matrices. They developed algebraic methods that significantly advanced the understanding of structures related to unimodular rows and their completion (see \cite {RAO}, \cite{RAOS}). In this paper, we use Suslin matrices to explore construction of new classes of linear codes.

The present work also makes use of the construction technique introduced by B. Hurley and T. Hurley in their paper (\cite{HurleyHurley2014}), ``Systems of MDS codes from units and idempotents." Specifically, we rely on the methodology discussed in the section on codes from units and adapt it to our setting. The codes thus obtained are however not MDS.

In this work, we also explore a different and relatively less-studied direction, the construction of linear block codes from recursive matrix families, particularly those arising from Suslin's construction. Suslin matrices are characterized by their recursive definitions and rich algebraic properties, offering a natural framework for generating structured codes with predictable parameters and potentially good error-correcting performance.

In particular, we investigate a family of codes $\mathcal{SRC}_r$ which we call as Suslin-Rao codes defined using Suslin matrices over $\mathbb{F}_2$. 
The construction presented in this paper leads to a class of linear block codes with explicitly computable parameters. In particular, the codes obtained have length $n=2^r$, dimension $k=2^{r-1},$ and minimum distance $d=r$ or $r+1$ depending on whether $r$ is even or odd positive integer.

Furthermore, we compare these codes with some standard constructions using the minimum distance and bit-error rate as performance metrics. 
The principal contributions of this paper are:
\begin{enumerate}
    \item We introduce a recursive framework for constructing generator matrices of linear block codes based on Suslin matrices.
    \item We explicitly determine the code parameters $(n, k, d)$ of the resulting family of codes based on the recursion index.
    \item We investigate structural and algebraic properties of the constructed codes, including self-duality conditions and their behavior with respect to classical coding-theoretic bounds.
    \item We provide a comparative performance analysis between the proposed codes and selected classical linear codes.
\end{enumerate}

The remainder of the paper is organized as follows: 
Section~2 presents the preliminaries and notation used throughout the paper. 
Section~3 provides a brief overview of Suslin matrices and their fundamental properties. 
In Section~4, we introduce the proposed Suslin-Rao codes, including their construction, basic properties, decoding approach and weight distribution. 
Section~5 presents a comparative analysis of the proposed code with several classical and modern linear block codes in terms of parameters, code rate, relative distance and bit error rate (BER) performance. 

\section{Preliminaries and Notations} 
Let us denote by $\mathbb{F}_2$, the finite field over $2$ elements and by $\mathbb{F}_2^n$, the $n-$ dimensional vector space over $\mathbb{F}_2$. We recall some basic definitions from (\cite{RR}). The Hamming weight of $c\in\mathbb{F}_2^n$ denoted by $wt(c)$ is the number of non-zero coordinates of $c.$ For any vectors $x,y \in \mathbb{F}_2^n$, the Hamming distance between $x$ and $y$ denoted by $d(x,y)$ is the number of coordinates in which $x$ and $y$ differ. We call $C$ a linear code over $\mathbb{F}_2$ if it is a linear subspace of $\mathbb{F}_2^n$ and elements of $C$ are called as codewords of $C.$ We denote by $d,$ the minimum Hamming distance between any two distinct codewords of $C.$ The support of a codeword $c = (c_1, \dots, c_n) \in C$ is the set of all indices at which $c$ is nonzero i.e., 
$\operatorname{supp}(c) = \{i \in \{1, \dots, n\} : c_i \ne 0\}$.

A generator matrix of $C$ is a $k\times n$ matrix $G\in M_{k,n}(\mathbb{F}_2),$ whose rows form a basis of $C$ i.e., $C=\{xG~|~x\in \mathbb{F}_2^k\}$. A parity check matrix of $C$ is a $(n-k)\times n$ matrix $H\in M_{(n-k),n}(\mathbb{F}_2)$ such that $C=\{c\in \mathbb{F}_2^n~|~cH^T=0\}$. The generator matrix $G$ and the parity check matrix $H$ are connected through the orthogonality relation $GH^T=0$, which reflects the fact that the row space of $G$ is the orthogonal to the row space of $H.$

The dimension of a linear code $C$, denoted by $k$ is the dimension of $C$ as a vector space over $\mathbb{F}_2^n$ and is also coincides with the rank of any of its generator matrices. Throughout this paper, a linear code over $\mathbb{F}_2$ with length $n$, dimension $k$ and minimum Hamming distance $d$ is referred to as an $[n,k,d]_2$ code.

The dual of code $C \subseteq \mathbb{F}_2^n$ denoted by $C^\perp$ is defined as the set of vectors orthogonal to every codeword in $C$
i.e., $
C^\perp = \{ y \in \mathbb{F}_2^n \mid x\cdot y = 0 \text{ for all } x \in C \}.
$
If $C$ is linear code then its dual code $C^\perp$ is also a linear code and has dimension $n-k$. A linear code $C \subseteq \mathbb{F}_q^n$ is said to be self-dual if it is equal to its dual, i.e.,
$C = C^\perp.$
This implies that $C$ must have dimension $k = n/2$, so self-dual codes exist only when $n$ is even (see \cite{FN}, \cite{Huff}).

We also recall here various bounds on linear codes. A fundamental result in coding theory that, any $[n,k,d]_2$ linear code satisfies the Singleton bound: $d \leq n - k + 1.$ This bound provides an upper bound on the minimum distance of a linear code in terms of its length and dimension. Codes that meet this bound with equality, i.e., $d = n - k + 1$, are called Maximum Distance Separable (MDS) codes. Examples of MDS codes include Reed-Solomon codes and trivial repetition or parity-check codes (see \cite{FN}, \cite{Huff}). The Griesmer bound gives a lower bound on the length of a linear code as a function of its dimension and minimum distance, and is known to be tighter than the Singleton bound for many binary codes. In particular, for a  linear code with parameters $[n,k,d]_2,$ the length must satisfy $\displaystyle n\geq \sum_{i=0}^{k-1} \lceil\frac{d}{2^i}\rceil$ (see \cite{Griesmer}). The sphere packing bound is an upper bound on the number of codewords a code can have for a given length and minimum distance. For a binary $[n,k,d]$ linear code, the sphere packing bound states that $\displaystyle 2^k \sum_{i=0}^t {n \choose i}\leq 2^n,$ where $t=\lfloor \frac{d-1}{2}\rfloor$ is the number of correctable errors (see \cite{FN}). 

The rate of a code given by $R=\frac{k}{n}$, represents the proportion of information symbols relative to the total number of transmitted symbols. A higher code rate corresponds to greater transmission efficiency, whereas a lower rate generally allows for improved error-correction capability (see \cite{RR}).

This foundational theory forms the basis of many powerful code families including Hamming, Simplex, BCH, Hadamard and Reed-Muller codes. In this paper, we extend this framework by constructing generator matrices using structured algebraic matrices, such as Suslin matrices, to define new families of binary linear codes. Thus, the classical algebraic K-Theory developed in $1970'$s has application in coding theory.
\section{Suslin Matrices} 
A. A. Suslin (see \cite{AS}) constructed square matrices of order $2^r$ denoted by  $S_r(v,w)$ using a recursive process. These matrices were later referred to as Suslin matrices. For a positive integer $r$ and vectors $v=(a_0, a_1,\dots,a_{r}),$ $w=(b_0, b_1,\dots,b_{r}) \in M_{1,r+1}(R)$ over any commutative ring R with unity with $det(S_r(v,w))=\langle v,w\rangle^{2^{r-1}}$, Suslin defined  $$S_r(v,w) =
\begin{pmatrix}
   a_0I_{2^{r-1}} & S_{r-1}(v_1,w_1)\\
  -S_{r-1}(w_1,v_1)^T & b_0I_{2^{r-1}}
                                                                                                                                                                                                                                      \end{pmatrix}$$
where $v_1=(a_1,a_2,\dots,a_{r})$ and $w_1=(b_1,b_2,\dots,b_{r})$. We call these matrices as Suslin matrices with respect to $v,w$.\\
To begin with we set $S_0(a_0,b_0)=\begin{pmatrix} a_0 \end{pmatrix}$. Then, we have $$S_1((a_0,a_1),(b_0,b_1))=
\begin{pmatrix}
 a_0 & a_1 \\
 -b_1 & b_0 
 \end{pmatrix};$$
, $$S_2((a_0,a_1,a_2),(b_0,b_1,b_2))=
\begin{pmatrix}
 a_0 & 0 & a_1 & a_2 \\
 0 & a_0 & -b_2 & b_1\\
 -b_1 & a_2 & b_0 & 0\\
 -b_2 & -a_1 &0 & b_0
 \end{pmatrix};
$$
$$S_3((a_0,a_1,a_2,a_3),(b_0,b_1,b_2,b_3))=
\begin{pmatrix}
 a_0 & 0 & 0 & 0 & a_1 & 0 & a_2 & a_3\\
 0 & a_0 & 0 & 0 & 0 & a_1 & -b_3 & b_2\\
 0 & 0 & a_0 & 0 & -b_2 & a_3 & b_1 & 0\\
 0 & 0 & 0 & a_0 & -b_3 & -a_2 &0 & b_1\\
 -b_1 & 0 & a_2 & a_3 & b_0 & 0 & 0 & 0\\
 0 & -b_1 & -b_3 & b_2 & 0 & b_0 & 0 & 0\\
 -b_2 & a_3 & -a_1 & 0 & 0 & 0 & b_0 & 0\\
 -b_3 & -a_2 &0 & -a_1& 0 & 0 & 0 & b_0  
 \end{pmatrix}.$$
 \label{def:2.1}

Note that $S_r(v, w)$ has entries only $0,\pm a_i,$ and $\pm b_i$ $(0 \leq i \leq r)$ 
and each row and column has at most $r + 1$ non-zero elements.
We recall below, one of the important property of the Suslin matrix. For a detailed study on Suslin matrices, the interested reader should refer to the paper (see \cite{RAO}).

\begin{thm} (see \cite{AS})  \label{thm:1.1}
 For $r \geq 1,$ and $v = (a_0, a_1, . . . , a_r), w = (b_0, b_1, . . . , b_r) \in M_{1,r+1}(R)$, the Suslin matrix $S_r(v, w)$ has the following property: 
 $$S_r(v,w)S_r(w,v)^T=\langle v,w\rangle I_{2^r}=S_r(w,v)^T S_r(v,w).$$ 
\end{thm}
We now discuss the special particular case of Suslin matrices $S_r(v,w)$ over the binary field $\mathbb{F}_2$ with respect to pair $(v,w)$ such that $v$ and $w$ are all $1$ vectors, i.e., $v_{1,r+1}=(1,1,\dots,1)$ and $w_{1,r+1}=(1,1,\dots,1)$. Clearly, the inner product of $v$ and $w$ in $\mathbb{F}^{r+1}_2$ is either $0$ or $1$ i.e.,
$\langle v,w\rangle= 
\begin{cases}
0, & \mbox{if }r~odd,\\
1, & \mbox{if }r~even.  
\end{cases}
$\\ 

Below are the examples of special particular case of $S_r(v,w)$ for $r=1,2$ and $3$.

$$S_1((1,1),(1,1))=
\begin{pmatrix}
 1 & 1 \\
 1 & 1 
 \end{pmatrix};~
S_2((1,1,1),(1,1,1))=
\begin{pmatrix}
 1 & 0 & 1 & 1 \\
 0 & 1 & 1 & 1\\
 1 & 1 & 1 & 0\\
 1 & 1 &0 & 1
 \end{pmatrix}
;$$
$$S_3((1,1,1,1),(1,1,1,1))=
\begin{pmatrix}
 1 & 0 & 0 & 0 & 1 & 0 & 1 & 1 \\
 0 & 1 & 0 & 0 & 0 & 1 & 1 & 1\\
 0 & 0 & 1 & 0 & 1 & 1 & 1 & 0\\
 0 & 0 & 0 & 1 & 1 & 1 &0 & 1\\
 1 & 0 & 1 & 1 & 1 & 0 & 0 & 0 \\
 0 & 1 & 1 & 1 & 0 & 1 & 0 & 0 \\
 1 & 1 & 1 & 0 & 0 & 0 & 1 & 0\\
 1 & 1 & 0 & 1& 0 & 0 & 0 & 1  
\end{pmatrix}.$$

From hereon, we denote the special particular case of Suslin matrix i.e., $S_r(\bf{1,1})$ by $\mathcal{S}_r$, where ${\bf{1}}$ is all one vector in ${\mathbb F}_2^{r+1}$. 

Note that $\mathcal{S}_r$ is a matrix of order $2^r$ and has only $0$ and $1$ as its entries and contains exactly $r+1$ ones in every row and column. 
Also note that $\mathcal{S}_r$ can be recursively constructed as follow:\\ 
$\mathcal{S}_r=\begin{pmatrix}
      I_{2^{r-1}} & \mathcal{S}_{r-1}\\
      \mathcal{S}_{r-1} & I_{2^{r-1}}
     \end{pmatrix}
$, where $I_{2^{r-1}}$ is the identity matrix of order $2^{r-1}$.\\
For example, $\mathcal{S}_2=\begin{pmatrix}
                  I_2 & \mathcal{S}_1\\
                  \mathcal{S}_1 & I_2
                  \end{pmatrix}
             ,~~$
             $\mathcal{S}_3=\begin{pmatrix}
                  I_{2^2} & \mathcal{S}_2\\
                  \mathcal{S}_2 & I_{2^2}
                  \end{pmatrix}
             ~~$ etc.\\
\\ 
We now study some properties of these special particular case of Suslin matrices.

\begin{lemma} \label{lem:1.2}
 $\mathcal{S}_r$ is symmetric for all $r$.
 \end{lemma}
\begin{proof}
 We prove this result by induction on $r$. For $r=1$, $\mathcal{S}_1=\begin{pmatrix}
   1 & 1\\
  1 & 1                                                                                                                                                                                                                                \end{pmatrix},$ which is clearly a symmetric matrix.
  Further we assume that the result is true for $r=k-1$, i.e., $\mathcal{S}_{k-1}^T=\mathcal{S}_{k-1}$ and we prove the result for $r=k$. By definition,
  \begin{align*}
   \mathcal{S}_k^T &=                                                                                                                                                                                                                                   \begin{pmatrix}
   I_{2^{k-1}} & \mathcal{S}_{k-1}\\
  \mathcal{S}_{k-1} & I_{2^{k-1}}                                                                                                                                                                                                                          \end{pmatrix}^T 
  =                                                                                                                                                                                                                                   \begin{pmatrix}
   I_{2^{k-1}} & \mathcal{S}_{k-1}^T\\
  \mathcal{S}_{k-1}^T & I_{2^{k-1}}                                                                                                                                                                                                                          \end{pmatrix}
  =                                                                                                                                                                                                                                   \begin{pmatrix}
   I_{2^{k-1}} & \mathcal{S}_{k-1}\\
  \mathcal{S}_{k-1} & I_{2^{k-1}}                                                                                                                                                                                                                          \end{pmatrix} 
  =\mathcal{S}_k.
  \end{align*}
\end{proof}
From Theorem, \ref{thm:1.1} it follows that, $det(\mathcal{S}_r)=1$ when $r$ is even and is equal to zero when $r$ is odd. Therefore, $\mathcal{S}_r$ is non-singular whenever $r$ is even. 

\begin{thm}
$\mathcal{S}_r$ is nilpotent of index $2$ whenever $r$ is odd and it is an involution whenever $r$ is even.
\label{thm:1.4}
\end{thm}
\begin{proof}
 By construction, each row and each column of $\mathcal{S}_r$ has $r+1$ ones, which implies $\mathcal{S}_r\neq {\bf{0}}.$ 
 By Theorem $\ref{thm:1.1}$, $\mathcal{S}_r\mathcal{S}_r^T=\langle {\bf{1}},{\bf{1}}\rangle I_{2^r}$. Using the fact that $\mathcal{S}_r$ is symmetric from Lemma \ref{lem:1.2}, we write, $\mathcal{S}_r\mathcal{S}_r=\langle {\bf{1}},{\bf{1}}\rangle I_{2^r}$
  i.e., $\mathcal{S}_r^2=\langle {\bf{1}},{\bf{1}}\rangle I_{2^r}.$ 
 If $r$ is odd, then $\langle {\bf{1}},{\bf{1}}\rangle=0$ and if $r$ is even, then $\langle {\bf{1}},{\bf{1}}\rangle=1.$ Hence,
 $$\mathcal{S}^2_r =\begin{cases} \mbox{ }0, & \mbox{if }r \mbox{ } is~odd,  \\ \mbox{ }1, & \mbox{if } r \mbox{ }is~even.  
         \end{cases}
$$

It follows that $\mathcal{S}_r$ is nilpotent of index $2$ when $r$ is odd and $\mathcal{S}_r$ is an involution matrix whenever $r$ is even.
\end{proof}
 
From Theorem \ref{thm:1.4}, 
it follows that for odd $r,$ \{$\mathcal{S}_r\}$ is a family of nilpotent matrices of index $2$.
\begin{rem} \label{rem2}
 As $\mathcal{S}_r^2=I_{2^r}$ whenever $r$ is even it follows that any two distinct rows of $\mathcal{S}_r$ are orthogonal (i.e., their dot product is zero) over $\mathbb{F}_2$.
\end{rem}

\begin{lemma} \label{lem:1.3}
 For $r$ even, if $R_j$ is the $j$-th row of $\mathcal{S}_r$ and $e_j$ is the $j$-th standard basis vector in $\mathbb{F}_2^{\,r+1}$ then we have : $$\displaystyle\sum_{i \in \operatorname{supp}(R_j)} R_i = e_j.$$
 
\end{lemma}
\begin{proof}
 Let $x \in \mathbb{F}_2^{\,r+1}$ be such that
$
\displaystyle\sum_{i \in \operatorname{supp}(R_j)} R_i = x.
$
Then the $k^{th}$ coordinate of $x$ is
$\displaystyle
x_k = \sum_{i \in \operatorname{supp}(R_j)} (R_i)_k,
$ where $(R_i)_k$ denotes the $k^{th}$ coordinate of row $R_i.$ 
Since $\mathcal{S}_r$ is symmetric for all $r$ (see Lemma \ref{lem:1.2}), we have
$
(R_i)_k = (R_k)_i.
$ Now we make the following cases.\\
\noindent{Case 1:} $k = j$.

\noindent It follows that 
$
\displaystyle x_j = \sum_{i \in \operatorname{supp}(R_j)} (R_j)_i.
$
Since $(R_j)_i = 1,$ exactly for $i \in \operatorname{supp}(R_j)$. As each row of 
$\mathcal{S}_r$ has weight $r+1$ and $r$ is even we get
$\displaystyle
x_j = \sum_{i \in \operatorname{supp}(R_j)} (R_j)_i 
= r+1 \equiv 1 \pmod{2}.
$ As $x_j\in \mathbb{F}_2,~x_j=1.$

\noindent {Case 2:} $k \neq j$.

When $k \neq j$, only those indices $i$ such that 
$
i \in \operatorname{supp}(R_j) 
 ~\text{and}~
k \in \operatorname{supp}(R_i)
$ contribute in a nontrivial manner to the sum $\displaystyle
x_k = \sum_{i \in \operatorname{supp}(R_j)} (R_i)_k
$.
Symmetry of $\mathcal{S}_r$ implies that
$
k \in \operatorname{supp}(R_i)
~ \text{if and only if} ~
i \in \operatorname{supp}(R_k).
$ 
Therefore, $x_k = \left| \operatorname{supp}(R_j) \cap \operatorname{supp}(R_k) \right|.$
From Theorem \ref{thm:1.4} and Remark \ref{rem2}, distinct rows of $\mathcal{S}_r$ have even dot product for even $r$. Therefore, $\left| \operatorname{supp}(R_j) \cap \operatorname{supp}(R_k) \right|$ is even.  Hence, $x_k \equiv 0 \pmod{2},$ so $x_k = 0$ as $x_k\in \mathbb{F}_2.$ 

Thus, $x = e_j.$ This completes the proof.
\end{proof}

\begin{thm}  \label{thm:1.3}
Let $\mathcal{S}_r=\begin{pmatrix}
A\\B                                                                                \end{pmatrix},$ 
where $A$ and $B$ are block sub-matrices of $\mathcal{S}_r$ of size $(2^{r-1}\times 2^r)$ each. Then for $r$ odd, each row of block $B$ is a linear combination of exactly $r$ rows of block $A$ and vice versa.
\end{thm}
\begin{proof}
From the construction of special particular case of Suslin matrix, $\mathcal{S}_r=\begin{pmatrix}
                  I_{2^{r-1}} & \mathcal{S}_{r-1}\\
                  \mathcal{S}_{r-1} & I_{2^{r-1}}
                 \end{pmatrix} =\begin{pmatrix}
A\\B                                                                                \end{pmatrix}$, where $A=\begin{pmatrix}I_{2^{r-1}} & \mathcal{S}_{r-1} \end{pmatrix}$ and $B=\begin{pmatrix} \mathcal{S}_{r-1} & I_{2^{r-1}} \end{pmatrix}$.

For $1\leq i \leq 2^{r-1}$, the $i^{th}$ row vector of $A$ is of the type $R'_i=(e_i,R_i)$ and that of $B$ is of the type $R''_i=(R_i,e_i)$, where $e_i$ is $i^{th}$ row of the identity matrix $I_{2^{r-1}}$ and $R_i$ is the $i^{th}$ row of $\mathcal{S}_{r-1}.$

For each $j$, $1\leq j \leq 2^{r-1}$, the $j^{th}$ row of $\mathcal{S}_{r-1}$ can be written as $\displaystyle R_j=\sum_{i\in Supp(R_j)} e_i$. From Lemma \ref{lem:1.3}, if $r$ is odd then $j^{th}$ row of $I_{2^{r-1}}$ can be written as  $\displaystyle e_j=\sum_{i\in Supp(R_j)} R_i.$ Therefore, $j^{th}$ row of $B,$  $\displaystyle R''_j=(R_j,e_j)=\left (\sum_{i\in Supp(R_j)} e_i~,~ \sum_{i\in Supp(R_j)} R_i \right ).$ This implies that the $j^{th}$ row of $B$ is a sum of as many rows of $A$ as the size of $Supp(R_j)$. 
As $\mathcal{S}_{r-1}$ has $r$ ones in each row, size of $Supp(R_j)$ is always $r$ for all $j.$
Hence, each row of block $B$ is linear combination of exactly $r$ rows of block $A$.
Similarly we can prove that each row of block $A$ is linear combination of exactly $r$ rows of block $B$.
This completes the proof.
\end{proof}

 \section{Suslin-Rao Codes}
 
 We recall a construction from the paper of Paul Hurley and Ted Hurley in (see \cite{PT1}, \cite{PT2}), where $(n,r)$ linear code was constructed from units in a commutative ring $R$ with unity as follows: 
let $U=\begin{pmatrix}
A\\B                                                                                \end{pmatrix}
,$ where $A$ is $r\times n$ matrix over $R$ and $B$ is $(n-r)\times n$ matrix over $R$. Also, let $V=\begin{pmatrix}
C & D                                                                                \end{pmatrix},$ 
where $C$ is $n\times r$ and $D$ is $n\times (n-r)$ matrix. Suppose $UV=I$ in the set $M_n(R)$ of matrices over a commutative ring $R$ with unity. Note $UV=I$ implies $AD=0$. A code constructed by considering the matrix $A$ as a generator matrix will be a $(n,r)$ code with check matrix $D^T$. Note that we can take $r$ as long as $1\leq r \leq n-1$ (If $r=0$, then $A$ is trivial and does not define a code and if $r=n$, then $B$ and $D$ do not exist, making $U$ and $V$ trivial).

We take this idea further for the case when $UV={0}$, where ${0}$ is the zero matrix. In this case, we again have $AD=0$, but a fundamental difference is that the product $UV$ does not give the identity matrix. It suggests a situation where all row space of $U$ is mapped to zero by $V$. This implies that all rows of $A$ are mapped to zero by 
$D$, which means that $A$ is not necessarily full-rank. If $A$ is not full-rank, it cannot generate a valid $(n,r)$-code because the number of independent codewords may be less than $r$. However if rows of $A$ are linearly independent,  then we can expect to get a $(n,r)$-linear code.

\subsection{Construction of Suslin-Rao Codes} \label{sec:construction}
We now use the above technique to construct linear codes from the Suslin matrices. 
We consider special particular case of Suslin matrix $\mathcal{S}_r$. 
From Theorem \ref{thm:1.1}, we have
$$\mathcal{S}_r\mathcal{S}_r^T= \begin{cases} \mbox{ }0, & \mbox{if }r~is~odd,  \\ 
\mbox{ }I_{2^r}, & \mbox{if } r ~is~even.  
         \end{cases}
$$\\
We now divide $\mathcal{S}_r=\begin{pmatrix}
A\\B                                                                                \end{pmatrix}
$ and divide $\mathcal{S}_r^T=\begin{pmatrix}
C & D                                                                                \end{pmatrix}$, where $A$ is $k\times n$ matrix, $B$ is $(n-k)\times n$ matrix, $C$ is $n\times k$ matrix and 
$D$ is $n \times (n-k)$ matrix. Note that $n=2^r$ and we keep $k$ strictly equal to $2^{r-1}$. We have $\mathcal{S}_r\mathcal{S}_r^T$ is either $\bf{0}$ or $I_{2^r},$ in both cases $AD=0.$ From the recursive definition of $\mathcal{S}_r$, $A=(I_{2^{r-1}}~|~\mathcal{S}_{r-1})$. It is easy to see that rows of matrix $A$ are linearly independent. We use matrix $A$ as generator matrix and corresponding matrix $D^T=(\mathcal{S}_{r-1}~|~I_{2^{r-1}})$ as parity check matrix to generate $(n,k)$ linear code. Hence forth we will call this code as Suslin-Rao code and we denote it by $\mathcal{SRC}_r.$ Clearly $\mathcal{SRC}_r$ is $(n,k)=(2^r,2^{r-1})$ linear code with code rate $k/n=1/2$. 

Let us illustrate this with the help of examples.
Below are the examples of $\mathcal{SRC}_r$ for $r=1,2.$\\
For $r=1$, $\mathcal{S}_1=
\begin{pmatrix}
 1 & 1 \\
 1 & 1 
 \end{pmatrix}
$,
thus $G_1=A=\begin{pmatrix}
 1 & 1
 \end{pmatrix}
$ and $\mathcal{SRC}_1=\{11,00\}$.\\
For $r=2$, 
$\mathcal{S}_2=
\begin{pmatrix}
 1 & 0 & 1 & 1 \\
 0 & 1 & 1 & 1\\
 1 & 1 & 1 & 0\\
 1 & 1 &0 & 1
 \end{pmatrix}
$, thus $G_2=A=\begin{pmatrix}
 1 & 0 & 1 & 1 \\
 0 & 1 & 1 & 1
 \end{pmatrix}=\begin{pmatrix}
 I_{2^1}~|~\mathcal{S}_1
 \end{pmatrix}
$ and 
$\mathcal{SRC}_2=\{1011,0111,1100,0000\}$.

Similarly, one can generate Suslin-Rao code $\mathcal{SRC}_r$ with generator matrix $G_r=\begin{pmatrix}                                               I_{2^{r-1}}~|~\mathcal{S}_{r-1}                                                            \end{pmatrix}.$

\begin{thm}(Minimum Hamming distance of the Suslin-Rao Code)\label{thm:2.1}\\
 The minimum Hamming distance of Suslin-Rao code $\mathcal{SRC}_r$ is: \\
 $$d_{min}(\mathcal{SRC}_r) =\begin{cases} \mbox{ }r+1, & \mbox{if }r \mbox{ } odd,  \\ \mbox{ }r & \mbox{if } r \mbox{ }even.  
         \end{cases}
$$ 
\end{thm}
\begin{proof}
 We prove this result in two parts. In first part, we prove $d_{min}(\mathcal{SRC}_r)=r+1$ for odd $r$ and in second part we prove $d_{min}(\mathcal{SRC}_r)=r$ for even $r.$  

 \noindent{Case I:} $r$ is odd.\\
 We prove the result by induction.  When $r=1$, the generator matrix for $\mathcal{SRC}_1$ is $G_1=\begin{pmatrix} 1 & 1 \end{pmatrix}.$ Codewords generated from $G_1$ are $\{11,00\}.$ Clearly the minimum Hamming distance of $\mathcal{SRC}_1$ is $d_{min}(\mathcal{SRC}_1)=2$. Thus, the result is true for $r=1.$ We now assume that the result is true for an odd integer $r=k>1$ i.e., $d_{min}(\mathcal{SRC}_k)=k+1$ and we prove $d_{min}(\mathcal{SRC}_{k+2})=k+3.$

From the construction of Suslin-Rao code, generator matrix of $\mathcal{SRC}_{k+2}$ is given by $G_{k+2}= A_{k+2} = [I_{2^{k+1}} \mid \mathcal{S}_{k+1}].$ Therefore, codewords of $\mathcal{SRC}_{k+2}$ are of the form,  
$c =xG_{k+2}=x[I_{2^{k+1}} \mid \mathcal{S}_{k+1}]=(x,\ x \mathcal{S}_{k+1}) ~\text{for } x \in \mathbb{F}_2^{2^{k+1}},\ x \neq 0.
$

Let us fix some nonzero $x$. Then, $wt(c) = wt(x) + wt(x S_{k+1}).$ We now make cases on $wt(x)$. \\
\noindent {Case 1: } We assume $x=e_i$, for some $i$ i.e., $wt(x) = 1$.\\
Then $x S_{k+1}$ is the $i$-th row of $S_{k+1}$. As each row of $S_{k+1}$ has weight $k + 2$, 
$
wt(c) = 1 + (k + 2) = k + 3.
$ 
So, at least one codeword in $\mathcal{SRC}_{k+2}$ has weight $k + 3$.
\\
{Case 2: ${wt}(x) \geq 2.$}\\
Write $x=(x_1,x_2)\in \mathbb{F}_2^{2^k} \times \mathbb{F}_2^{2^k}.$ Therefore, 
$
x \mathcal{S}_{k+1} = (x_1,\ x_2)
\begin{pmatrix}
I_{2^k} & \mathcal{S}_k \\
\mathcal{S}_k & I_{2^k}
\end{pmatrix}
= \left(x_1 + x_2 \mathcal{S}_k,\ x_1 \mathcal{S}_k + x_2\right)
.$

\noindent We now make the following subcases:\\
\noindent{$(a)$} Both $x_1$ and $x_2$ are non-zero i.e., $x_1\neq 0$ and $x_2\neq 0.$
We claim that at least one of $x_1 + x_2 \mathcal{S}_k$ and $x_1 \mathcal{S}_k + x_2$ is nonzero. If both are zero i.e., $x_1 + x_2 \mathcal{S}_k=0$ and $x_1 \mathcal{S}_k + x_2=0$ then from first equation we get $x_1=x_2 \mathcal{S}_k$. From the second equation we get $(x_2 \mathcal{S}_k)\mathcal{S}_k+x_2=0$. It follows that, $x_2(\mathcal{S}_k^2+I)=0$. Since $k$ is odd, $\mathcal{S}_k^2=0$. Hence, $x_2=0$, a contradiction. 
Note that, each of the vectors $x_1 + x_2 \mathcal{S}_k$ and $x_1 \mathcal{S}_k + x_2$ is a codeword in $\mathcal{SRC}_k$. As $d_{min}(\mathcal{SRC}_k)=k+1,$ every non-zero codeword in $\mathcal{SRC}_k$ has weight at least $k + 1$. Hence, $
wt(x_1 + x_2 \mathcal{S}_k) \geq k + 1~{or}~wt(x_1 S_k + x_2) \geq k + 1.
$
This implies,
$$
wt(x S_{k+1}) \geq k + 1
\text{ and therefore, }wt(c) \geq 2 + (k + 1) = k + 3.
$$

\noindent(b) Exactly one of $x_1$ and $x_2$ is zero. Let us say $x_1=0$ and $x_2\neq 0$.\\
Clearly,  
$
x S_{k+1} = (x_2 S_k,\ x_2) \text{ which implies } wt(x S_{k+1}) = wt(x_2 S_k) + wt(x_2).$
As $x_2 \ne 0$ and $\mathcal{SRC}_k$ has minimum weight $k + 1$, we have
$
wt(x_2 S_k) + wt(x_2) > k + 1.$ This implies 
$
wt(c) = wt(x) + {wt}(x S_{k+1}) > 2 + (k + 1) = k + 3
.$

Hence, for all $x$ with ${wt}(x) > 1$, we have:
$wt(c) \geq k + 3.$ 
As $d_{min}=\displaystyle\min_{c\in C} wt(c)$ (see \cite{RR}), it is clear that minimum Hamming distance $d_{min}(\mathcal{SRC}_{k+2})=k+3$. Hence, $d_{\min}(\mathcal{SRC}_r) = r+1$ for all odd $r.$
\vspace{2mm}

\noindent{Case II: $r$ is even.}\\
Consider the Suslin-Rao code $\mathcal{SRC}_{r}$ with generator matrix $G_{r}.$ By Construction  
$$G_{r}=\begin{pmatrix}                                                                                                                             I_{2^{r-1}}~|~\mathcal{S}_{r-1}                                                                                                                        \end{pmatrix}=\begin{pmatrix}
I_{2^{r-1}}~|~\begin{pmatrix}                                                                                                                                  I_{2^{r-2}} & \mathcal{S}_{r-2}\\                                                                                                                        \mathcal{S}_{r-2} & I_{2^{r-2}}                                                                                                                                  \end{pmatrix}
\end{pmatrix}=\begin{pmatrix}
I_{2^{r-1}} ~|~ \begin{pmatrix}                                                                                                                                  A\\                                                                                                                        B                                                                                                                                  \end{pmatrix}
\end{pmatrix}.$$
The block submatrix $A=\begin{pmatrix} I_{2^{r-2}} & \mathcal{S}_{r-2} \end{pmatrix}$ of $\mathcal{S}_{r-1}$ is a generator $G_{r-1}$ of Suslin-Rao code $\mathcal{SRC}_{r-1}$.
As $r-1$ is odd, $d_{min}(\mathcal{SRC}_{r-1})=r$ and from Theorem \ref{thm:1.3}, every row vector of block submatrix $B$ is a linear combination of rows of $A=G_{r-1}$.  Thus, codewords generated from  $\begin{pmatrix}
I_{2^{r-1}} &|& \begin{pmatrix}                                                                                                                                  A\\                                                                                                                        B                                                                                                                                  \end{pmatrix}
\end{pmatrix}$ can not have minimum Hamming weight less than $r$. 
Now we check whether $\mathcal{SRC}_{r}$ has a codeword of Hamming weight equal to $r$.

From Theorem \ref{thm:1.3} we know, each row of $B$ is linear combination of $r-1$ rows of $A.$ Lets say row vectors $\{R_{A,1},\dots,R_{A,r-1}\}$ of matrix $A$ generates row vector $R_{B,i}$ of matrix $B$. Addition of these rows together contributes a zero vector. When we add corresponding rows of the matrix $G_{r}=\begin{pmatrix}                                                                                                                                  I_{2^{r-1}} &|& \mathcal{S}_{r-1}                                                                                                                        \end{pmatrix}=\begin{pmatrix}
I_{2^{r-1}} &|& \begin{pmatrix}                                                                                                                                  A\\                                                                                                                        B                                                                                                                                  \end{pmatrix}
\end{pmatrix},$ i.e., row vectors $\{G_1,\dots,G_{r-1}\}$ and $G_i$ we get a row vector of weight exactly $r$ ($r$ rows of first block of $G_{r}$ i.e., identity matrix contributes $r$ ones and second block contributes zero in this weight). This is the codeword in $\mathcal{SRC}_{r}$ with weight $r$. Therefore, $d_{min}(\mathcal{SRC}_{r})=\displaystyle\min_{c\in \mathcal{SRC}_r} wt(c)=r.$ 
This completes the calculation of minimum distance of $\mathcal{SRC}_r$ for $r$ even.                                                                                     
\end{proof}

By construction $\mathcal{SRC}_r$ is a linear code with $(n,r)=(2^r,2^{r-1})$. By Theorem \ref{thm:2.1}, for odd $r$, $\mathcal{SRC}_r$ is $[2^r,2^{r-1},r+1]$ linear block code and whenever $r$ is even $\mathcal{SRC}_r$ is $[2^r,2^{r-1},r]$ linear code.
\begin{rem}
 The Suslin-Rao code $\mathcal{SRC}_3$, is  equivalent to extended Hamming code $[8,4,4,]_2$ (As they have identical generator matrices).
\end{rem}

\subsection{Some Basic Properties of Suslin-Rao Codes}
\begin{prop}
 For $r$ odd, the Suslin-Rao code $\mathcal{SRC}_r$ is self-dual.
\end{prop}
\begin{proof}
From the construction of Suslin-Rao codes discussed in Section \ref{sec:construction}, generator matrix of the code $\mathcal{SRC}_r$ is $G_r=A=\begin{pmatrix} I_{2^{r-1}} & S_{r-1}\end{pmatrix}$  and the corresponding parity check matrix is $H_r=D^T=\begin{pmatrix}S_{r-1} & I_{2^{r-1}}                  \end{pmatrix}$.  $G_r$ and $H_r$ are both $2^{r-1}\times 2^r$ matrices with full row rank. Therefore, $dim(G_r)=dim(H_r)=2^{r-1}.$ This is equal to $\frac{n}{2}$ as, $n=2^r.$  

From (\cite{Huff}) and (\cite{FN}), a linear code $C$ is self-orthogonal if and only if $GG^T=0$, where $G$ is generator matrix of $C$. From Theorem \ref{thm:1.1} and Theorem \ref{thm:1.4}, for odd $r$, $\mathcal{S}_r^2=0$. It follows that, $AA^T=G_rG_r^T=0$. Therefore, all codewords in $\mathcal{SRC}_r$ are  orthogonal to each other.  


Now we prove that the sets of codewords spanned by $G_r$ and $H_r$ are equal. For that we look into the recursive construction of special particular case of Suslin matrices, \\
$\mathcal{S}_r=\begin{pmatrix}A\\B \end{pmatrix}=\begin{pmatrix}
I_{2^{r-1}} & \mathcal{S}_{r-1}\\
\mathcal{S}_{r-1} & I_{2^{r-1}}
\end{pmatrix}$ and 
$\mathcal{S}_r^T=\begin{pmatrix}
C & D                                                                                \end{pmatrix}=\begin{pmatrix}
I_{2^{r-1}} & \mathcal{S}_{r-1}\\
\mathcal{S}_{r-1} & I_{2^{r-1}}
\end{pmatrix}.$\\
Clearly, 
the check matrix $H_r=D^T$, of code $\mathcal{SRC}_r$ is equal to block submatrix $B$ of $\mathcal{S}_r$.
By Theorem \ref{thm:1.3}, each row of matrix $H_r$ is a linear combination of rows of $G_r$ and vice versa whenever $r$ is odd. Hence both matrices spans the same set of codewords. Thus, $\langle G_r \rangle=\langle H_r\rangle$. This proves that for odd $r,$ Suslin-Rao code $\mathcal{SRC}_r$ is a self-dual code.
\end{proof}

\begin{prop}
 The Suslin-Rao code $\mathcal{SRC}_r$ satisfy the Singleton bound.
\end{prop}
\begin{proof}
It is well known that any linear code with parameters $[n,k,d]$ satisfy the Singleton bound $d \leq n-k+1$ (see \cite{RR}).
 It gives upper bound on minimum distance $d$.
 We claim that, the minimum distance $d$ of Suslin-Rao code $\mathcal{SRC}_r$ satisfy this bound. 
 
 By construction, Suslin-Rao code $\mathcal{SRC}_r$ has $n=2^r$, $k=2^{r-1}$ and $d=r$ if $r$ is even and $d=r+1$ if $r$ is odd. Therefore, $n-k+1=2^r-2^{r-1}+1=
  2^{r-1}+1.$ It is easy to see that, in both cases i.e., for odd and even $r,$  $d\leq 2^{r-1}+1.$  Therefore, $d\leq n-k+1$ is always true for Suslin-Rao code $\mathcal{SRC}_r$ for all $r$. 
This proves that Suslin-Rao code satisfy the Singleton bound. 

\end{proof}

\begin{prop}
 The Suslin-Rao code $\mathcal{SRC}_r$ satisfy the Griesmer bound.
 
\end{prop}
\begin{proof}
 It is well known from (\cite{Griesmer}) that, every $[n,k,d]$ linear code satisfy Griesmer bound i.e., $\displaystyle n\geq \sum_{i=0}^{k-1} \lceil\frac{d}{2^i}\rceil.$ We claim that, the length of the Suslin-Rao code also satisfy this bound. 
  By construction, Suslin-Rao code $\mathcal{SRC}_r$ has $n=2^r$, $k=2^{r-1}$ and $d=r$ if $r$ is even and $d=r+1$ if $r$ is odd. 
Therefore, when $r$ is even, 
\begin{align*}
 \displaystyle \sum_{i=0}^{k-1} \lceil\frac{d}{2^i}\rceil &= \sum_{i=0}^{k-1} \lceil\frac{r}{2^i}\rceil
 = \lceil\frac{r}{1}\rceil+\lceil\frac{r}{2^1}\rceil+\dots+\lceil\frac{r}{2^{k-1}}\rceil. 
\end{align*}

Here, if $2^i\leq r,$ then $\lceil\frac{r}{2^i}\rceil\leq \frac{r}{2^i}+1$ and if $2^i> r,$ then $\lceil\frac{r}{2^i}\rceil=1$. Let us assume $t$ is the largest integer power index such that $2^t\leq r$ i.e., $t\leq \lfloor log_2(r)\rfloor$ and then we split the right-hand sum over at $t$. 
Therefore,  
\begin{align*}
 \displaystyle \sum_{i=0}^{k-1} \lceil\frac{d}{2^i}\rceil 
 &=\sum_{i=0}^{t} \lceil\frac{r}{2^i}\rceil+\sum_{i=t+1}^{k-1} \lceil\frac{r}{2^i}\rceil
 =\sum_{i=0}^{t} \lceil\frac{r}{2^i}\rceil+\sum_{i=t+1}^{k-1} 1\\
 &\leq \left(r+\frac{r}{2}+\dots+\frac{r}{2^t}+(t+1) \right )+(k-1-t)\\
 &\leq r\left(1+\frac{1}{2}+\dots+\frac{1}{2^t} \right )+(t+1)+(k-1-t)\\
 &< 2r+k=2r+2^{r-1}.
 \end{align*}
Clearly, for all even $r>2$, $\displaystyle n=2^r> 2r+2^{r-1}\geq\sum_{i=0}^{k-1} \lceil\frac{d}{2^i}\rceil$. 
For $r=2$, the parameters of $\mathcal{SRC}_r$ are given by $n = 2^2 = 4$, $k = 2^{2-1} = 2$, and $d = 2$. Evaluating the exact Griesmer sum yields:
\begin{align*}
\sum_{i=0}^{k-1} \left\lceil \frac{d}{2^i} \right\rceil &= \sum_{i=0}^{1} \left\lceil \frac{2}{2^i} \right\rceil
= \left\lceil \frac{2}{1} \right\rceil + \left\lceil \frac{2}{2} \right\rceil= 2 + 1 = 3.
\end{align*}
Since $n = 4 > 3$, the Griesmer bound holds true for $r=2$.
Using similar arguments, when $r$ is odd we obtain $\displaystyle n>2(r+1)+2^{r-1}\geq\sum_{i=0}^{k-1} \lceil\frac{d}{2^i}\rceil$ for all $r>3.$ 
One can easily check Griesmer bound holds, by evaluating the exact Griesmer sum for cases $r=1$ and $3$. 
Hence, the Suslin-Rao code $\mathcal{SRC}_r$ satisfy the Griesmer bound.
\end{proof}

\begin{prop}
 The Suslin-Rao code $\mathcal{SRC}_r$ satisfy the sphere packing bound.
\end{prop}
\begin{proof}
 From (\cite{RR}), sphere packing bound for $[n,k,d]$ linear code over $\mathbb{F}_2$ is given by $\displaystyle 2^k \sum_{i=0}^t {n \choose i}\leq 2^n$, where $t=\lfloor \frac{d-1}{2}\rfloor.$ 
 First we find $t$, if $d=r$ then $t=\lfloor \frac{r-1}{2}\rfloor$ and if $d=r+1$ then $t=\lfloor \frac{r}{2}\rfloor.$ In both cases $t\leq \frac{r}{2}$. For Suslin-Rao code $\mathcal{SRC}_r,$ $n=2^r$, $k=2^{r-1},$ thus it is enough to prove that $\displaystyle\sum_{i=0}^t {2^r \choose i}\leq 2^{2^{r-1}}.$
 As ${n \choose i} \leq n^i,$ for all integers $i>0$, we have 
 \begin{align*}
  \displaystyle\sum_{i=0}^t{2^r \choose i}\leq \sum_{i=0}^t(2^r)^i&=\frac{(2^r)^{t+1}-1}{2^r-1}\\
  &< (2^r)^{t+1}\\
  &\leq{2^r}^{(\frac{r}{2}+1)}\\
  &<2^{2^{r-1}}~\text{for all~}r\geq 6. 
 \end{align*}
The remaining cases i.e., $r=1$ to $5$ can be verified manually. Hence, the Suslin-Rao code $\mathcal{SRC}_r$ satisfy the sphere packing bound. 
\end{proof} 

\subsection{Decoding the Suslin-Rao Code}

For decoding the proposed Suslin-Rao code, we adopt the classical syndrome decoding framework for linear codes (see \cite{FN}). Let $C$ be a linear $[n,k]$ code over $\mathbb{F}_q$ with parity-check matrix $H \in \mathbb{F}_q^{(n-k)\times n}$. If a codeword $x \in C$ is transmitted and the received vector is $y=x+e$, where $e \in \mathbb{F}_q^n$ denotes the error vector, then the associated syndrome is defined by
$s = Hy^T.$
The syndrome uniquely determines the coset $e+C$, which contains all vectors producing the same syndrome.

Decoding therefore consists of selecting an error vector $\hat e$ from this coset and forming the estimate
$
\hat x = y - \hat e.
$
In classical syndrome decoding, this selection is typically performed using a coset-leader rule, where the vector of smallest Hamming weight in the corresponding coset is chosen (see \cite{RR}). In the simulations and decoding procedure for the proposed Suslin-Rao code, we employ this standard syndrome-based decoding approach.


\subsection {Weight Distribution of Suslin-Rao Code}
Weight distribution is one of the fundamental parameter of a linear code. It captures the Hamming weight profile of the code. The weight distribution of a code $C$ is the sequence $\{A_w\}_{w=0}^n$, where $A_w$ denotes the number of codewords in $C$ of Hamming weight $w$ i.e., 
$A_w = |\{c \in C : \ {wt}(c) = w\}|,$ where $wt(c)$ is the Hamming weight of codeword $c$ (see \cite{FN}). In Suslin-Rao code $\mathcal{SRC}_r$, for all $r,$ we have $A_0=1$. If $r$ is odd then $A_w=0$ for $0<w<r+1$ and if $r$ is even then $A_w=0$ for $0<w<r.$ This information is compactly encoded in the weight enumerator, where the weight enumerator of $C$ is the polynomial $\displaystyle W_C(z) = \sum_{w=0}^n A_w z^w$ (see \cite{RR}). Weight enumerators of Suslin-Rao code $\mathcal{SRC}_r$ for $1\leq r\leq 5$ are as follows.
{\scriptsize
{
\begin{align*}
W_{\mathcal{SRC}_1}(z) &= 1 + z^2. \\
W_{\mathcal{SRC}_2}(z) &= 1 + z^2 + 2z^3. \\
W_{\mathcal{SRC}_3}(z) &= 1 + 14z^4  + z^{8}. \\
W_{\mathcal{SRC}_4}(z) &= 1 + 14z^4 + 8z^5 + 24z^6 + 56z^7 + 49z^8 + 56z^9 + 28z^{10} + 8z^{11} + 8z^{12} + 4z^{14}. \\
W_{\mathcal{SRC}_5}(z) &= 1 + 32z^6 + 300z^8 + 1952z^{10} + 6976z^{12} + 14400z^{14} + 18214z^{16} + 14400z^{18}+ 6976z^{20} + 1952z^{22} + 300z^{24} + 32z^{26} + z^{32}.
\end{align*}  
}}

Following result demonstrates symmetry in weight distribution of a linear code containing all one codeword.

\begin{prop} \label{pro:5}
 If $C$ is a binary linear  code of length $n$ and contains the all one vector as a codeword, then the weight distribution of $C$ is symmetric i.e., $A_w=A_{n-w}$ for all $0\leq w\leq n$. 
\end{prop}
\begin{proof}
 Define mapping $\phi$ from $C$ to $C$ as $\phi(c)=c+{\bf 1}.$ As ${\bf 1}\in C$ and $C$ is linear, $\phi$ is a bijection from $C$ to itself. Moreover, $wt(\phi(c))=n-wt(c).$ Thus, $\phi$ induces a bijection between codewords of weight $w$ and those of weight $n-w$. Hence, $A_w=A_{n-w}$ for all $0\leq w\leq n$.
\end{proof}

Symmetry described in Proposition \ref{pro:5}, can be clearly seen in $\mathcal{SRC}_r$ for only odd $r$ (check weight distributions above). Note $\mathcal{SRC}_r$ does not contain all $\bf{1}$ vector when $r$ is even.

\section{Comparative Analysis of Code Performance}

\subsection{Comparison of Parameters} 
 We compare the parameters of Suslin-Rao code $\mathcal{SRC}_r$ for $r=3$ to $7$ with several well-known classical codes of similar length, including \textit{Hamming codes}, \textit{Simplex codes}, \textit{Hadamard codes}, \textit{Reed--Muller codes} and \textit{BCH codes} (see \cite{FN}). The comparisons are shown in Tables~\ref{tab:comparison 1} and~\ref{tab:table 2}.

 \begin{table}[htbp]
\centering
\small
\renewcommand{\arraystretch}{1.3}
\begin{tabular}{|c|c|c|c|c|}
\hline
{$r$} & {Suslin-Rao $\mathcal{SRC}_r$} & {Hamming} & {Simplex } & {Hadamard} \\
\hline
3 & [8, 4, 4] & [7, 4, 3] & [7, 3, 4] & [8, 3, 4]\\
\hline
4 & [16, 8, 4] & [15, 11, 3] & [15, 4, 8] & [16, 4, 8]\\
\hline
5 & [32, 16, 6] & [31, 26, 3] & [31, 5, 16] & [32, 5, 16]\\
\hline
6 & [64, 32, 6] & [63, 57, 3] & [63, 6, 32] & [64, 6, 32]\\
\hline
7 & [128, 64, 8] & [127, 120, 3] & [127, 7, 64] & [128, 7, 64]\\
\hline
\end{tabular}
\caption{Comparison of the Suslin-Rao Code $\mathcal{SRC}_r$ with Hamming, Simplex, and Hadamard Codes}\label{tab:comparison 1}
\end{table}

\begin{table}[htb]
\centering
\small
\renewcommand{\arraystretch}{1.3}
\begin{tabular}{|c|c|c|c|}
\hline
\( r \) & Suslin-Rao $\mathcal{SRC}_r$ & \( \mathcal{RM}(1, r) \) or \( \mathcal{RM}(r-2, r) \) & BCH Code \\
\hline
3 & [8, 4, 4] & [8, 4, 4] & [7, 4, 3] \\
\hline
4 & [16, 8, 4] & [16, 5, 8] or [16, 11, 4] & [15, 5, 7] or [15, 7, 5] \\
\hline
5 & [32, 16, 6] & [32, 6, 16] or [32, 26, 4] & [31, 11, 11] or [31, 16, 7] \\
\hline
6 & [64, 32, 6] & [64, 7, 32] or [64, 57, 4] & [63, 10, 27] or [63, 18, 21] \\
\hline
7 & [128, 64, 8] & [128, 8, 64] or [128, 120, 4] & [127, 36, 31] or [127, 64, 21] \\
\hline
\end{tabular}
\caption{Comparison of the Suslin-Rao Code $\mathcal{SRC}_r$ with Reed-Muller ( $\mathcal{RM}$(r, m)) and BCH codes}\label{tab:table 2}
\end{table}

From the tables, we observe the following:
\begin{enumerate}
    \item For each value of $r$, Suslin-Rao code $\mathcal{SRC}_r = [2^r, 2^{r-1}, d]$ achieves a good trade off between code rate and minimum Hamming distance. In particular, it improves on the error correction capability (i.e., minimum distance) compared to Hamming codes of similar length.

    \item For $r = 3$, the Suslin-Rao code $[8, 4, 4]$ has a higher minimum distance than the Hamming code $[7, 4, 3]$, and matches the Simplex code in terms of distance. The number of codewords (i.e., $2^4 = 16$) is larger than that of the Simplex code, which has only 8.

    \item For $r = 4$, the Suslin-Rao code $[16, 8, 4]$ has better minimum distance than the Hamming code $[15, 11, 3]$, and while the Simplex code $[15, 4, 8]$ has higher minimum distance, it achieves that at the cost of a much lower dimension. 

    \item For $r = 5$, the Suslin-Rao code $[32, 16, 6]$ outperforms the Hamming code $[31, 26, 3]$ in both minimum  distance and structure. 

    \item For $r = 6$ and $r = 7$, similar trends continue. The Suslin-Rao codes maintain a minimum distance of 6 and 8 respectively, which is significantly better than the constant minimum distance of 3 in Hamming codes. 

    \item The Reed-Muller code $\text{RM}(r, m)$ comes in two variants depending on whether one considers the code or its dual. For instance, for $r = 5$, it may appear as $[32, 6, 16]$ or $[32, 26, 4]$. The former has excellent distance but low dimension, and the latter has high dimension but poor error correction. In contrast, the Suslin-Rao code $[32, 16, 6]$ offers a balanced middle ground.

    \item Hadamard codes typically have very high minimum distances (e.g., $d = 16$ for length $32$), but come with extremely low dimensions (e.g., $k = 5$). Hence, while powerful for error detection/correction, they are inefficient in terms of data rate. The Suslin-Rao code outperforms Hadamard codes in terms of code rate, and provides sufficient minimum distance for moderate noise channels.
    
    \item BCH codes provide a flexible trade-off between code rate and minimum distance by adjusting the designed error-correcting capability. For $r = 5$, binary BCH codes like $[31, 11, 11]$ or $[31, 16, 7]$  are available. These codes can outperform the Suslin-Rao code in terms of distance but often at the cost of reduced dimension. The Suslin-Rao code $[32, 16, 6]$, while slightly inferior in distance to $[31, 16, 7]$, offers comparable dimension, making it attractive for low-complexity applications.
\end{enumerate}

\noindent This comparison can be summarized as follows: Suslin-Rao codes $\mathcal{SRC}_r$ provide, a balanced trade-off between distance and dimension, better minimum distance than Hamming at comparable lengths, higher code rate than Simplex and Hadamard codes and  Reed-Muller codes.

These comparisons demonstrate that the proposed Suslin-Rao codes are strong candidates for practical applications where moderate error correction and high throughput are desired, especially in systems where decoding simplicity is important.

\subsection{Rate and Relative Distance Comparison}
We compare the code rate $R = \frac{k}{n}$ and relative distance $ \delta = \frac{d}{n}$ of the Suslin-Rao code $\mathcal{SRC}_r$ with classical codes including \textit{Hamming code}, \textit{Simplex code}, \textit{Reed-Muller code} and \textit{Hadamard code}. The comparisons are shown in Table~\ref{tab:table 3}.

\begin{table}[h]
\label{tab:table 3}
\centering
\begin{tabular}{|c|c|c|}
\hline
\textbf{Code} & \textbf{Rate $R$} & \textbf{Relative Distance $\delta$} \\
\hline
Suslin-Rao $\mathcal{SRC}_r$ & $\frac{1}{2}$ & $\frac{r+1}{2^r}$ (odd $r$) or $\frac{r}{2^r}$ (even $r$)\\
\hline
Hamming & $1 - \frac{r}{2^r - 1}$ & $\frac{3}{2^r - 1}$ \\
\hline
Simplex & $\frac{r}{2^r - 1}$ & $\frac{2^{r-1}}{2^r - 1}$ \\
\hline
Reed-Muller $\mathcal{RM}(1, r)$ & $\frac{r+1}{2^r}$ & $\frac{2^{r-1}}{2^r}$ \\
\hline
Hadamard & $\frac{r}{2^r}$ & $\frac{1}{2}$ \\
\hline
\end{tabular}
\caption{Asymptotic Comparison of Code Parameters} 
\end{table}
The Suslin-Rao code $\mathcal{SRC}_r$ strikes a balanced trade-off between rate and relative distance. It offers higher rate than Simplex codes and comparable or better relative distance than Reed--Muller codes of the same order.

\subsection{BER Performance Comparison}

To examine the practical error-correcting performance of the proposed Suslin-Rao code, we consider a standard digital transmission model over an additive white Gaussian noise (AWGN) channel. The encoded binary sequences are first modulated using Binary Phase Shift Keying (BPSK), where each bit is mapped to a real-valued symbol. Specifically, the bits $0$ and $1$ are mapped to $+1$ and $-1$, respectively. BPSK modulation is widely used in performance analysis due to its analytical tractability and robustness.

The modulated symbols are transmitted over an AWGN channel (see \cite{John}). In this model, independent Gaussian noise samples are added to each transmitted symbol. The received signal is therefore given by
$ y = x + n,$
where $x \in \{+1,-1\}$ denotes the transmitted symbol and $n$ is a zero-mean Gaussian random variable with variance determined by the signal-to-noise ratio.

Performance is quantified in terms of the Bit Error Rate (BER), defined as the ratio of the total number of incorrectly decoded bits to the total number of transmitted information bits. The BER is evaluated across a range of $E_b/N_0$ values from $0$ dB to $10$ dB in increments of $1$ dB. Here $E_b$ denotes the average energy per information bit and $N_0$ represents the one-sided noise power spectral density of the AWGN channel. The ratio $E_b/N_0$ therefore measures the signal-to-noise ratio per transmitted bit and is commonly used to characterize the reliability of digital communication systems.

The proposed Suslin-Rao code $[16,8,4]$ was benchmarked against several well known linear block codes of comparable block length, namely:
Hamming $[15,11,3]$, Simplex $[15,4,8]$, BCH $[15,7,5]$, Hadamard $[16,4,8]$, and Reed-Muller $[16,5,8]$ codes(see \cite{FN}). 

For each $E_b/N_0$ value, $20{,}000$ independent message blocks were transmitted. In order to ensure a fair comparison among codes of different rates, the noise variance was adjusted according to the code rate $R = k/n$, so that
$
N_0 = \frac{1}{R \cdot (E_b/N_0)}.
$
This normalization guarantees that all codes are compared under identical energy-per-information-bit conditions.

At the receiver, hard-decision decoding was employed for all codes. The received BPSK symbols were first converted into binary values via threshold detection, assigning $0$ to positive samples and $1$ to negative samples.

For the Suslin-Rao code, syndrome decoding was implemented using its parity check matrix to detect and correct single bit error patterns. The Hamming code was decoded using minimum distance decoding. The BCH code was decoded using MATLAB's built-in \texttt{comm.BCHDecoder} object. The Simplex, Hadamard, and Reed Muller codes were decoded using exhaustive minimum distance (brute force) decoding by comparing the received vector with all valid codewords and selecting the closest match in Hamming distance.

The BER curves plotted in Figure \ref{fig:ber-comparison}, illustrate the comparative performance of the considered codes over the AWGN channel. Key observations from the results are summarized below: 
\begin{figure}[h] 
    \centering
    \includegraphics[width=0.65\textwidth]{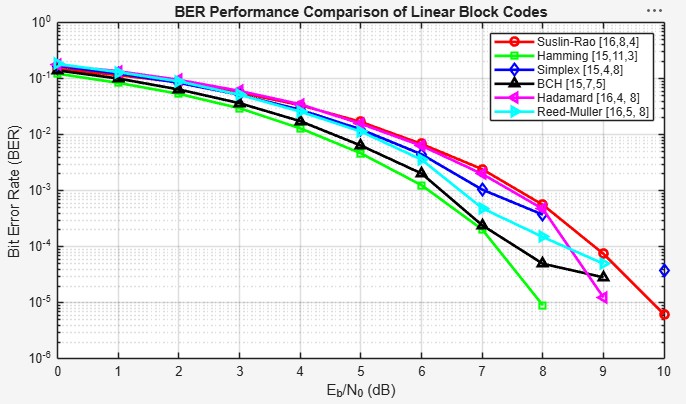} 
    \caption{BER Comparison of Suslin-Rao Code}
    \label{fig:ber-comparison}
\end{figure}

\subsubsection*{Key Observations and Quantitative Comparison}

The BER curves in Figure~\ref{fig:ber-comparison} exhibit the expected exponential decay behavior characteristic of linear block codes over the AWGN channel under hard-decision decoding. The performance comparison may be interpreted in three distinct SNR regimes.

\paragraph{Low SNR Region (0-3 dB):}
In the noise dominated region, all codes demonstrate comparable BER performance. The minimum distance structure does not significantly influence decoding reliability in this regime, and the curves remain closely clustered.

\paragraph{Moderate SNR Region (4-7 dB):}
Performance separation becomes clearly visible as the signal to noise ratio increases. At a BER level of approximately $10^{-3}$, the proposed Suslin-Rao $[16,8,4]$ code requires roughly $7.5$ dB. In comparison, the Hamming $[15,11,3]$ code achieves the same BER near $6$ dB, while the BCH $[15,7,5]$ code requires approximately $6.5$ dB. The Simplex $[15,4,8]$ code reach at this BER level at $7$ dB, while Reed Muller $[16,5,8]$ code requires approximately $6.5$ dB. Hadamard $[16,4,8]$ requires approximately $7.5$ dB which is very close to Suslin-Rao code.  

Thus, at BER $\approx 10^{-3}$, the Suslin-Rao code exhibits performance comparable to classical block codes of similar length, with differences within approximately $0.5$ dB to $1.5$ dB.

\paragraph{High SNR Region (8-10 dB):}
In the high-SNR regime, the asymptotic slope of the BER curve becomes the dominant performance indicator. The proposed $[16,8,4]$ Suslin-Rao code continues to demonstrate stable exponential decay, reaching BER values on the order of $10^{-5}$ at approximately $10$ dB. 

The observed slope is consistent with the minimum distance $d=4$, indicating that the dominant error events are governed by weight-$4$ codewords. The decay behavior remains regular and does not exhibit early flattening within the simulated SNR range.

\paragraph{Effect of Code Rate:} 
It is important to emphasize that the considered codes operate at different rates:
\[
R_{\text{Suslin-Rao}} = \frac{1}{2}, ~
R_{\text{Hamming}} = \frac{11}{15}, ~
R_{\text{Simplex}} = \frac{4}{15}, ~
R_{\text{BCH}} = \frac{7}{15}, ~
R_{\text{RM}} = \frac{5}{16}.
\]
Although rate normalization was incorporated in the noise variance computation, higher-rate codes inherently transmit more information per channel use, which influences practical system efficiency.

In this context, the Suslin-Rao code achieves a balanced trade off between rate and error performance. Despite operating at rate $1/2$, it maintains BER characteristics comparable to lower-rate constructions such as Reed Muller and Simplex codes under identical hard decision decoding assumptions.

\section{Acknowledgement}
The corresponding author (Anuradha S. Garge) would like to express her sincere gratitude to Professor Manoj Kumar Keshari (IIT Bombay) for organizing in $2018,$ the Advanced Instructional School on Projective Modules at IIT Bombay, where she got the idea of exploring the possibility of constructing codes from Suslin matrices. She also thanks Professor Ravi A. Rao for introducing her to the subject of $K-$theory specially to the papers of Suslin and Vaserstein. This paper was unfortunately not seen by him due to his sad demise on May $5$, $2026$. With deep gratitude, we dedicate this paper to his memory.

\end{document}